\documentclass[ba,preprint]{imsart}

\RequirePackage{amsthm,amsmath,amsfonts,amssymb}
\RequirePackage[numbers]{natbib}
\RequirePackage[colorlinks,citecolor=blue,urlcolor=blue,backref=page,backref=page]{hyperref}
\RequirePackage{graphicx}

\pubyear{2024}
\volume{TBA}
\issue{TBA}
\firstpage{1}
\lastpage{1}

\startlocaldefs

\usepackage{mathrsfs}
\usepackage{bbm}

\usepackage{trimspaces}
\makeatletter
\newcommand*{\trim}[1]{%
  \trim@spaces@noexp{#1}%
}
\makeatother

\newcommand{\NN}{\mathbb{N}}
\newcommand{\PP}{\mathbb{P}}
\newcommand{\RR}{\mathbb{R}}

\newcommand{\cC}{\mathcal{C}}

\newcommand{\cI}{\mathcal{I}}
\newcommand{\cM}{\mathcal{M}}

\newcommand{\cP}{\mathcal{P}}
\newcommand{\cQ}{\mathcal{Q}}
\newcommand{\cR}{\mathcal{R}}

\newcommand{\cY}{\mathcal{Y}}

\newcommand{\sP}{\mathscr{P}}
\newcommand{\sT}{\mathscr{T}}
\newcommand{\sX}{\mathscr{X}}
\newcommand{\sY}{\mathscr{Y}}

\DeclareMathOperator*{\argmax}{arg\,max}

\let\eps\varepsilon
\let\to\longrightarrow
\newcommand{\indic}{\mathbbm{1}}

\newcommand{\conv}[2][\ ]{\overset{#1}{\underset{#2}{\to}}}

\newcommand{\aseq}[1]{\underset{#1}{=}}
\newcommand{\equi}[1]{\underset{#1}{\sim}}
\newcommand{\mbf}[1]{\mathbf{#1}}

\theoremstyle{plain}

\newtheorem{thm}{Theorem}[section]

\newtheorem{prop}[thm]{Proposition}

\theoremstyle{definition}
\newtheorem{defi}[thm]{Definition}

\newtheorem{rem}[thm]{Remark}
\newtheorem{nota}[thm]{Notation}

\newtheorem{assu}[thm]{Assumption}
\theoremstyle{remark}

\endlocaldefs

\usepackage{lineno}
\modulolinenumbers[5]

\begin{document}

\begin{frontmatter}
\title{Properly constrained reference priors decay rates for efficient and robust posterior inference}
\runtitle{Properly constrained reference priors}

\begin{aug}
\author[A]{\fnms{Antoine}~\snm{Van Biesbroeck}\ead[label=e1]{antoine.van-biesbroeck@polytechnique.edu}}
\address[A]{CMAP, CNRS, École polytechnique and Université Paris-Saclay, CEA,  France\printead[presep={;\\ }]{e1}}

\runauthor{A. Van Biesbroeck}
\end{aug}

\begin{abstract}
In Bayesian analysis, reference priors are widely recognized for their objective nature. Yet, they often lead to intractable and improper priors, which complicates their application.
Besides, informed prior elicitation methods are penalized by the subjectivity of the choices they require. %
In this paper, we aim at proposing a reconciliation of the aforementioned aspects. Leveraging the objective aspect of reference prior theory, we introduce two strategies of constraint incorporation to build tractable reference priors.
One provides a simple and easy-to-compute solution when the improper aspect is not questioned, and the other introduces constraints to ensure the reference prior is proper, or it provides proper posterior.
Our methodology emphasizes the central role of Jeffreys prior decay rates in this process, and the practical applicability of our results is demonstrated using an example taken from the literature.
\end{abstract}



\end{frontmatter}


\section{Introduction}

Bayesian analysis offers a coherent framework for integrating prior information and updating beliefs with data. The incorporation of prior knowledge can be motivated by various needs, such as enhancing interpretability, introducing uncertainty, or embedding existing knowledge into the analysis. 
Actually, the selection of this \emph{a priori} centralizes a lot of attention in Bayesian inference,  and can be a critical aspect of the method. Indeed, it can significantly influence the posterior distribution and, consequently, the conclusions drawn from the analysis. Thus,  an inconsistent choice of prior can jeopardize the validity of the entire study, a challenge underscored by \citet{Gelman_2021}.

Various methods for prior elicitation have been developed to address this issue, offering different strategies for different purposes. Comprehensive reviews of them are proposed in \citep{Mikkola2023} and in \citep{Consonni2018}. Despite their variety, many of them involve subjective choices that remain open to criticism. 
Indeed, prior elicitation methods generally consist of choosing a way to inform a prior.  They are
often based on the selection of a class and hyperparameters that remain to be tuned.
Among the different challenges that are commonly addressed, we can emphasize the ease of posterior inference. 
Conjugate priors, for instance, are designed to provide a simple formulation of the posterior. Penalised Complexity (PC) priors \citep{Simpson2017} and sparse priors \citep{Castillo2015} are suitable for high-dimensional problems; and g-priors \citep{Liang2008} are used for variable selection in normal regression models. \\
In order to provide a better interpretability of the information transmission process, several studies favor a hierarchical design, in which the prior is the result of a latent prior modeling.
This approach allows the use 
of historical data, as in power priors \citep{Chen2006} and Information Matrix priors \citep{Gupta2009}; or even  imaginary data  \citep{Perez2002,Spitzner2011}.
Posterior priors, developed for sensitivity analysis studies \citep{Bousquet2023}, are also part of this  approach with imaginary data.
These hierarchical constructions are appreciated for their ability to tackle the improper aspect of most low-informative priors. Actually, they shift the focus to a ‘prior on the prior’ that must also be elicited.

When the context or the modeling requires %
a prior without
the possibility of incorporating uncritical beliefs, it is common to move towards  priors that are low-informative.
The reference prior theory arises in this case as a cornerstone. The theory provides a formal mechanism to incorporate prior information in a way that maximizes the information gained from the data within the issued \emph{a posteriori} quantities. This process ensures that the prior can be qualified as `objective'.
Extensively developed and formalized \citep{Bernardo1979a,Berger2009,Mure2018}, the reference priors are elected for their lack of subjectivity in various practical
studies \citep{Chen2008,Gu2016,Dandrea2021,VanBiesbroeck2023} and statistical models \citep{Paulo2005,Mure2019,Natarajan2000}. Yet, they are criticized for their  low computational feasibility.
As a matter of fact, the reference priors have been proven to generally lead to Jeffreys priors in asymptotic frameworks \citep{Clarke1994,Liu2014,VanBiesbroeckBA2023}, or hierarchical versions of them in  high-dimensional problems \citep{Bernardo1979,Bernardo2005,Gu2019}.
Despite their  objective nature,  their implementation is often cumbersome and not always recommended in high dimensions \citep{Berger2015}. %
Moreover, the low-informative nature of these priors is associated with their common improper aspect, necessitating careful handling to ensure valid statistical inference.

Thus, the construction of priors is expected to strike a balance between several criteria. While many works restrict their sets of priors to ones that are tractable, proper, or suitable for high dimensions, others seek to minimize any source of subjectivity.
This paper aims to reconcile all these criteria to improve the prior elicitation.
Our contribution takes the form of an enrichment of the reference prior theory to leverage the objective aspect that it provides to reference priors. Building on \citep{VanBiesbroeckBA2023},
we restrict priors to the ones that belong to  well-chosen ---and not too restrictive--- classes, and we introduce two strategies to define  convenient reference priors.
Our first strategy provides a simple,  tractable solution for constraining reference priors when the improper aspect is not questioned. The second, by contrast, introduces constraints that lead to  reference priors that are proper, or lead to proper posteriors. 
For both strategies, we try to define the potential loss of objectivity induced by the constraints and we discuss their limits.
Our results emphasize the central role of Jeffreys prior decay rates when they are improper.
Additionally, we draw attention to the fact that our methodology opens a way to define various reference priors on the basis of constraints that could result from any other motivation.

The work presented in this paper is motivated by practical considerations and based on rigorous theoretical foundations. In particular, we present a revisited Bayesian framework in Section \ref{sec:framework}. This section does not represent the main contribution of the paper (which is the derivation of reference priors under well-chosen constraints), but it does define the fundamental tools we use. 
It also proposes a complete modeling for Bayesian statistics that differs from usual Bayesian frameworks 
in that it systematically and consistently takes into account improper priors.
We suggest that readers particularly interested in the practical aspects of this work focus on the sections that follow, while readers intrigued by the detailed mathematical background are invited to explore that section. %
The rest of the paper begins with
a review of the reference prior theory and of the definitions of the reference priors in Section \ref{sec:refpriors}. It precedes the development  of the main contribution of the paper.
Indeed, Section \ref{sec:constraintsection} motivates our work and expresses our results about constrained reference priors, which are then discussed in Section \ref{sec:discussion}.
Eventually, the claimed practical aspect of our study is demonstrated by the application of our method to an example taken from  the literature in Section \ref{sec:example}. Detailed mathematical proofs are compiled in Section \ref{sec:proofs}, followed by a conclusion in Section \ref{sec:conclusion}.

\section{Bayesian framework and improper priors}\label{sec:framework}

The purpose of statistics is to observe real phenomena, and to construct their interaction with probability theory. Thus, since the early work of \citet{kolmogorov1933foundations}, it is appreciated to be reminded that there exists an abstract, underlying structure from which result the probabilistic description and tools that are used to perform inference. This structure, which takes the form of a probability space in the axiomatic of Kolmogorov, represents the indescribable phenomena that induce randomness in the observations.
In this section, we propose such a  construction for Bayesian statistics, in a way that allows priors to be improper.

The Bayesian framework is standardly defined from a collection of probability measures $(\PP_{Y|\theta})_{\theta\in\Theta}$ that constitutes a statistical model, alongside with a prior $\pi$ on $\Theta$ ---namely, a probability distribution---. It is therefore possible to construct a probability space $\mbf\Omega$ that models the observations: they are modeled as the independent realizations of a random variable $Y\in\cY$, with $Y$ being distributed according to $\PP_{Y|\theta}$ conditionally to $T=\theta$ where $T$ is a random variable with distribution  $\pi$.  

However, it is common when it resorts to Bayesian analysis to allow a construction based on an `improper' prior, namely, a prior defined from a $\sigma$-finite measure, whose total mass is not necessarily finite. 
Their consideration rests essential within the reference prior theory because %
the reference priors are known to be often improper.
It is not singular to allow such priors ---even though it is recommended that they are limits\footnote{Different authors in the literature propose different ways to define such limits; a review of them is not the purpose of this study.} of proper priors--- 
yet the usual Kolmogorov's axiomatics does not take  this possibility into account to construct the underlying space that provides the usual statistical modeling. 

Our suggestion for the construction of a Bayesian framework to model a Bayesian statistical system which allows the insertion of improper priors is based on the work of \citet{Taraldsen2016}, who adapted the propositions of \citet{renyi1970foundations}. Their idea is to extend the notion of a probability space when the probability is only a $\sigma$-finite measure. Up to a multiplicative constant, their framework is shown to coincide with standard probability spaces when the measure has a finite total mass, or when conditioning on a set with finite mass.

We let $(\Theta,\sT)$ be a measurable set. If $\pi_1$ and $\pi_2$ are two non-null, non-negative and $\sigma$-finite  measures we can define the relation $\simeq$ by
    \begin{equation}
        \pi_1\simeq\pi_2\Longleftrightarrow\exists\delta>0,\,\pi_1=\delta\pi_2.
    \end{equation}
Define by $\cM$ the space of non-null, non-negative and $\sigma$-finite measures on $\Theta$. The following properties are verified:
    \begin{itemize}
        \item The relation $\simeq$ is an equivalence relation on $\cM$. We note $[\pi]$ the class of a $\pi\in\cM$.
        \item Denote by $\cM^\nu$ the space of absolutely continuous measures w.r.t. $\nu\in\cM$. It is stable by $\simeq$.
        \item Define by $\cR$ the space of non-negative non-null measurable functions from $\Theta$ to $\RR$. The relation%
            \begin{equation}
                f\propto g \Longleftrightarrow\exists\delta>0,\,f=\delta g%
            \end{equation}
        is an equivalence relation on $\cR$. %
        Let $\nu\in\cM$, the Radon-Nikodym Theorem defines a natural surjection from the quotient space $\cR/\!\propto$ to $\cM^\nu/\!\simeq$. 
        A class in $\cR/\!\propto$ is called a density class of the unique class in $\cM^\nu/\!\simeq$ that it induces by this mapping. If two elements $[f]$, $[g]$ of $\cR/\!\propto$ are density classes of a same element in $\cM^\nu/\!\simeq$, they are equal $\nu$-a.e. in the sense $\exists\delta>0$, $f=\delta g$ $\nu$-a.e.
    \end{itemize}

\citet{Bioche2016} consider such a definition of $\sigma$-finite measures equal up-to-a constant to formalize the notion of prior and posterior approximations, whether they are proper or improper. Notably, they define a topology induced by a convergence in $\cM/\!\simeq$: the Q-vague convergence. %

A general prior is then defined as a class in $\cM/\!\simeq$.
The tuple $(\Theta,\sT,[\pi])$ corresponds to what \citet{Taraldsen2016} call a conditional probability space for any $[\pi]\in\cM/\!\simeq$.
For any $U \in \sT$ such that 
$\pi(U)\in (0,+\infty)$, 
that space issues a unique probability space on $U$, with the probability $\pi(\cdot|U)=\pi(\cdot\cap U)/\pi(U)$, it is independent of the choice of the representative of $[\pi]$.
This framework gives what is necessary to construct a general conditional probability space that models our problem, as we express below.

\begin{prop}\label{prop:kolmog}
We assume that the observations take values in a 
Polish space $(\cY,\sY)$ and are statistically modeled by the collection of conditional probabilities $(\PP_{Y|\theta})_{\theta\in\Theta}$ with $(\Theta,\sT)$ being a Polish space and such that $\forall A\in\sY$, $\theta\mapsto\PP_{Y|\theta}(A)$ is measurable. We consider $\pi\in\cM$.\\
Then there exist a conditional probability space $(\mbf\Omega, \mbf\Xi, [\mbf\Pi])$, a measurable process $\overline Y=(Y_i)_{i\in\NN}$, $Y_i:\mbf\Omega\to\cY$, and a measurable function $T:\mbf\Omega\to\Theta$ such that
    \begin{enumerate}
        \item for any $U\in\sT$, $\mbf\Pi(T\in U)=\pi(U)$; 
        \item if 
        $\pi(U) \in  (0,+\infty)$ 
        then the unique probability space conditioned on $\{T\in U\}$ is such that the conditional probability of any $(Y_{i_l})_{l=1}^k$ to $T=\theta$ is $\PP_{Y|\theta}^{\otimes k}$. More explicitly, if $A_{j_1},\dots,A_{j_k}\in\sY$ then
            \begin{equation}
                \mbf\Pi\Big(\bigcap_{l=1}^kY_{j_l}\in A_{j_l}|T\in U\Big) = \int_\Theta\PP_{Y|\theta}^{\otimes k}(A_{j_1}\times\dots\times A_{j_k})d\pi(\theta| U).
            \end{equation}
    \end{enumerate}
\end{prop}

Such existence of the underlying abstract space $\mbf\Omega$ provides the desired interpretability of the modeling. It can be seen as the source of the uncontrollable randomness that goes beyond the observations and the parameter. It is proven in Section \ref{sec:proofs}. One can notice that it is consistent with usual Bayesian frameworks when $\pi$ is proper.

\paragraph{Notations and assumptions for the remaining of the manuscript}
In this manuscript, the assumptions of Proposition \ref{prop:kolmog} are supposed to be verified. It will always be assumed that $\Theta\subset\RR^d$, for some $d\geq1$. It is equipped with the Borel $\sigma$-algebra denoted by $\sT$.
We restrain our study to the priors that admit locally bounded and a.e. continuous densities w.r.t. the Lebesgue measure on $\Theta$ denoted by $\varrho$. Indeed, the space $\cR^\cC$ of measurable, locally bounded and a.e. continuous functions from $\Theta$ to $\RR$ is stable by $\propto$, so that we can adopt the following notations in the manuscript.

\begin{nota}
    In this work `a prior' will refer to an element in $\cM^\varrho_\cC/\!\simeq$, the subspace in $\cM^\varrho/\!\simeq$ associated to $\cR^\cC/\!\propto$. 
    A class of priors is therefore a subset of $\cM^\varrho_\cC/\!\simeq$.
    In general, it will be common in this work to denote by $\pi$ a prior.  When it is proper, it will be natural to associate it with a representative of itself that has a total mass equal to $1$. Otherwise, we naturally extend some notations to $\pi$:
    \begin{itemize}
        \item[(i)] if $F:\cM^\varrho_\cC\to\RR$ we can define $F(\pi)$ as the subset in $\RR$ of evaluations of $F$ in the representatives of $\pi$. Then, the statement ``$F(\pi)=c$'' stands for ``$x=c\,\forall x\in F(\pi)$''.  
        Examples are $\pi(U)\subset[0,\infty]$ for a $U\in\sT$, or $\int_\Theta f(\theta)\pi(\theta)d\theta\subset[0,\infty]$ for a $f:\Theta\to[0,\infty)$ measurable (the latter expresses the set composed by the $\int_\Theta f(\theta)f_\pi(\theta)d\theta$ for any $f_\pi$ being a density of a representative of $\pi$). 
        \item[(ii)] if $f\in\cR^\cC$, we write $\pi(\theta)\propto f(\theta)$ to express that $\pi$ is the uniquely defined class of prior issued by the class of $f$ in $\cR^\cC/\!\propto$.         
        \item[(iii)] for any compact $U\subset\Theta$, we have $\pi(U)<\infty$ because the representatives of $\pi$ admit locally bounded densities. %
        The conditional probability $\pi(\cdot|U)$ on $U$ is called the re-normalized restriction of $\pi$ to $U$. %
    \end{itemize}
\end{nota}

The problem is assumed to admit a likelihood: there exist probability density functions $(\ell(\cdot|\theta))_{\theta\in\Theta}$ with respect to a common measure $\mu$ on $\cY$ such that
    \begin{align}
       \text{for $\rho$-a.e. $\theta$, } \forall A\in\sY,\, \PP_{Y|\theta}(A)=\int_A\ell(y|\theta)d\mu(y). %
    \end{align}

For given $\pi$, $k\geq1$ and $\mbf y\in\cY^k$, the posterior distribution is uniquely defined in  $\cM^{\varrho}/\!\simeq$ by its density $p(\cdot|\mbf y)$:

        \begin{align}\label{eq:posteriordef}
        &
        p(\theta|\mbf y) \propto\prod_{i=1}^k\ell(y_i|\theta)\pi(\theta)\propto\ell_k(\mbf y|\cdot)\pi(\cdot)  .   %
    \end{align}
The marginal distribution $\PP_{\mbf Y_k}$ is defined as the pushforward measure of $\mbf\Pi$ by $\mbf Y_k$. If the posterior in Equation (\ref{eq:posteriordef}) is proper for $\PP_{\mbf Y_k}$-a.e. $\mbf y$, the marginal distribution is uniquely defined by the density $p_{\mbf Y_k}$ w.r.t. $\mu^{\otimes k}$:
    \begin{equation}
        \forall\mbf y\in\cY^k,\,
        p_{\mbf Y_k}(\mbf y) = \int_\Theta p(\theta|\mbf y) d\theta = \int_\Theta\ell_k(\mbf y|\theta)\pi(\theta)d\theta.
    \end{equation}

Such a framework, under classical regularity assumptions of the likelihood w.r.t. $\theta$ lets the Fisher information matrix $\cI(\theta) = (\cI(\theta)_{i,j})_{i,j=1}^d$ be defined as:
\begin{equation}
 \cI(\theta)_{i,j} = - \int_\cY[\partial^2_{\theta_i\theta_j}\log\ell(y|\theta)]\, \ell(y|\theta)d\mu(y).
\end{equation}
The Jeffreys prior is defined in $\cM^{\varrho}/\!\simeq$ by its density $J$: $J(\theta)\propto\sqrt{|\det\cI(\theta)|}$.

\section{Reference prior definitions}\label{sec:refpriors}

In this section, we review the definitions of the reference priors, in order to suggest a suitable framework for the introduction of constraints.
This paper takes as a support the extended reference prior theory proposed in \citep{VanBiesbroeckBA2023}. In our work, we consider mutual information $I_{D_\alpha}$ defined by an $\alpha$-divergence, $\alpha\in(0,1)$, as a dissimilarity measure:
    \begin{align}
        &I_{D_\alpha}(\pi|k) = \int_\Theta D_\alpha(\PP_{Y|\theta}^{\otimes k}||\PP_{\mbf Y_k}) \pi(\theta)d\theta \label{eq:mutual}\\
        \nonumber\text{with}\quad 
            &D_\alpha(\PP_{Y|\theta}^{\otimes k}||\PP_{\mbf Y_k}) = \int_{\cY} f_\alpha\left(\frac{p_{\mbf Y_k}(\mbf y)}{\ell_k(\mbf y|\theta)}\right)\ell_k(\mbf y|\theta) d\mu^{\otimes k}(\mbf y),\nonumber %
    \end{align}
with $f_\alpha(x) = \frac{x^\alpha-\alpha x-(1-\alpha)}{\alpha(\alpha-1)}$. %
This quantity is ensured to have a sense when $\pi$ is proper, it will not be manipulated in other cases. %

The principle of the reference prior theory is to maximize w.r.t. $\pi$ the mutual information as expressed in Equation (\ref{eq:mutual}). %
The $D_\alpha$-reference priors are defined as:  %

\begin{defi}[$D_\alpha$-reference prior \citep{VanBiesbroeckBA2023}] \label{defi:refprior}
    Let $\cP$ be a class of priors. A prior $\pi^\ast\in\cP$ is called a $D_\alpha$-reference prior over the class $\cP$ with rate $\varphi(k)$ if there exists an openly increasing sequence of compact subsets $(\Theta_i)_{i\in \NN}$ with $\Theta_i\subset\Theta$, $\bigcup_{i\in \NN}\Theta_i=\Theta$, $\pi^\ast(\Theta_i)>0$ such that
        \begin{equation}\label{eq:limitmutinfo}
            \forall i\in \NN,\,\forall\pi\in\cP_i,\, \lim_{k\rightarrow\infty}\varphi(k)(I_{D_\alpha}(\pi_i^\ast|k)-I_{D_\alpha}(\pi_i|k)) \geq0,
        \end{equation}
    where $\pi_i$ and $\pi_i^\ast$ respectively denote the re-normalized restrictions of $\pi$ and $\pi^\ast$ to $\Theta_i$, and $\cP_i=\{\pi\in\cP,\,\pi(\Theta_i)>0\}$.
\end{defi}
We recall that the historical definition of reference priors made by \citet{Bernardo1979} is consistent with the above definition when the Kullback-Leibler divergence is used instead of the divergence $D_\alpha$ in the mutual information. Considering $D_\alpha$ can be seen as an extension of that historical definition, as it is known to tend toward the Kullback-Leibler divergence when $\alpha\to0$.

We also precise that some studies define the reference priors differently when the dimension of $\Theta$ is higher than $1$. Their definition corresponds to the hierarchical construction we present in Section \ref{sec:issuesJeff}.

Definition \ref{defi:refprior} actually slightly differs from the one in \citep{VanBiesbroeckBA2023} as we require here a property to be verified by the sequence $(\Theta_i)_{i\in\NN}$. A sequence $(\Theta_i)_{i\in \NN}$ is said to be openly increasing if there exists $i_0\geq0$ and a sequence $(V_i)_{i\geq i_0}$ of open subsets of $\Theta$ such that for any $i\geq i_0$:
    \begin{equation}
        \Theta_i\subset V_i\subset\Theta_{i+1}.
    \end{equation}

Below we recall a result \citep[Theorem 1]{VanBiesbroeckBA2023} which, under appropriate assumptions, gives an expression of the limit involved in Equation (\ref{eq:limitmutinfo}) when $\Theta$ is compact:

\begin{thm}[\citeauthor{VanBiesbroeckBA2023} \citeyear{VanBiesbroeckBA2023}]\label{thm:l(pi)}
    Suppose $\Theta$ to be compact and $\pi$ be a  prior. Then the $D_\alpha$-mutual information admits a limit:
        \begin{equation}
            \lim_{k\rightarrow\infty} k^{d\alpha/2} I_{D_\alpha}(\pi|k) = l(\pi), \qquad l(\pi) = C_\alpha \int_{\Theta}\pi(\theta)^{1+\alpha} |\cI(\theta)|^{-\alpha/2}  d\theta ,
        \end{equation}
        with $C_\alpha= (2\pi)^{d\alpha/2} (1-\alpha)^{-d/2}/(\alpha(\alpha-1))$.
\end{thm}

This theorem provides a clear asymptotic description of the $D_\alpha$-mutual information when $\Theta$ is compact. This way, it gives a necessary and sufficient condition for a prior to be a $
D_\alpha$-reference prior, as stated in the proposition below.

\begin{prop}\label{prop:}
    Let $\cP$ be a class of priors. $\pi^\ast\in\cP$ is a $D_\alpha$-reference prior over $\cP$ if and only if it exists an openly increasing sequence of compact sets $(\Theta_i)_{i\in\NN}$, $\pi^\ast(\Theta_i)>0$, $\bigcup_{i\in \NN}\Theta_i=\Theta$, such that for any $i\in \NN$, 
        \begin{equation}
            \pi^\ast_i\in\argmax_{\pi\in\cP_i}l(\pi_i);
        \end{equation}
    where $\cP_i=\{\pi\in\cP,\,\pi(\Theta_i)>0\}$ and with $\pi^\ast_i,\,\pi_i$ respectively denoting the re-normalized restriction of $\pi^\ast$ and $\pi$ to $\Theta_i$.
\end{prop}

In \citep{VanBiesbroeckBA2023}, the author proved that the Jeffreys prior satisfies the condition of above's proposition when $\cP$ is the set of priors admitting locally bounded and a.e. continuous densities (denoted by $\cM^\varrho_\cC/\!\simeq$ in section \ref{sec:framework}).

The choice of the class $\cP$ remains open and can be restrained from the large one of priors in $\cM^\varrho_\cC/\!\simeq$. %
However, such a restriction leaves really unsure the existence of a reference prior. 
Indeed, Definition \ref{defi:refprior} is itself restrictive, as to admit a reference prior, the class must contain a prior whose restrictions are optimal on any compact subsets of $\Theta$.
In this paper, we suggest an extension of the definition of reference priors in the case where in the class $\cP$, the optimal priors on compact subsets of $\Theta$ are not re-normalizations of each other, but converge to a  prior in $\cP$. Such convergence is considered in the sense of the Q-vague convergence \citep{Bioche2016} on $\cM^\varrho_\cC/\!\simeq$, as explained in the following definition.

\begin{defi}[Quasi reference prior]\label{defi:quasiRefprior}
    Let $\cP$ be a class of priors. We call $\pi^\ast\in\cP$ a quasi $D_\alpha$-reference prior if it exists an openly increasing sequence $(\Theta_i)_{i\in \NN}$  of compact sets with $\bigcup_{i\in\NN}\Theta_i=\Theta$ such that
    \begin{itemize}
        \item[(i)] for any $i\in \NN$, there exists a $D_\alpha$-reference prior $\pi_i^\ast$ over $\cP_i$, the set of re-normalized restrictions to $\Theta_i$ of priors in $\cP$,
        \item[(ii)] $\pi^\ast$ is the Q-vague limit of the sequence $(\pi^\ast_i)_{i\in \NN}$.
    \end{itemize}
\end{defi}

Proposition below ensures that this definition properly extends Definition \ref{defi:refprior}.
\begin{prop}\label{prop:quasi}
     \begin{itemize}
        \item If $\pi^\ast$ is a $D_\alpha$-reference prior over a class $\cP$, then it is a quasi $D_\alpha$-reference prior.
        \item If $\cP$ is a class of priors convex and stable by multiplication by indicator functions over measurable sets, then the quasi $D_\alpha$-reference prior is the unique $D_\alpha$-reference prior.
        \item If $\cP$ is a convex class of priors and if the sequence of subset $(\Theta_i)_i$ in Definition \ref{defi:quasiRefprior} is fixed, then the quasi-reference prior over class $\cP$ is unique.
    \end{itemize}
\end{prop}

\begin{proof}
    The first statement of the proposition is clear given the definition of a $D_\alpha$-reference prior.\\
    For the second, just notice that if $\Theta$ is compact and if $\pi^\ast$ is the maximal argument of $l$ over $\cP$, then its re-normalized restriction $\pi_1^\ast$ on a compact subset $U$ maximizes $l$ over the class of all re-normalized restrictions $\cP_U$.
    Indeed, if we suppose that $\pi_1\in\cP_U$ maximizes $l$ then, denoting $\pi_0^\ast$ the re-normalized restriction of $\pi^\ast$ to $\Theta\setminus U$, $t=\int_U\pi^\ast$, and $\pi = t\pi_1+(1-t)\pi_0^\ast$, $\pi\in\cP$ and
        \begin{equation}
            l(\pi) = t^{\alpha+1}l(\pi_1) + (1-t)^{\alpha+1}l(\pi_0^\ast) > t^{\alpha+1}l(\pi_1^\ast) + (1-t)^{\alpha+1}l(\pi_0^\ast) = l(\pi^\ast).
        \end{equation}
    Hence $\pi^\ast$ does not maximize $l$ over $\cP$, which is absurd.\\
    Therefore, in our problem, considering two sequences $(\pi^{(1)}_i)_i$ and $(\pi_i^{(2)})_i$ respectively defined on $(\Theta_i^{(1)})_i$ and $(\Theta^{(2)}_i)_i$, we will get that for any $i$, $\pi^{(1)}_i(\theta)=\pi_i^{(2)}(\theta)$ for all $\theta\in\Theta_i^{(1)}\cap\Theta_i^{(2)}$.\\
    Eventually, they are identical on every compact subsets of $\Theta$, equal to their Q-vague limits which are the same.\\
    Finally, the third statement of the proposition results from the strict concavity of $l$. Indeed for any $i$, the class $\cP_i$ of re-normalized priors on $\Theta_i$ is convex so that the maximal argument of $l$ over $\cP_i$ is unique. Hence the uniqueness of the quasi-reference prior over $\cP$.
\end{proof}

\section{Constrained reference priors}\label{sec:constraintsection}

\subsection{Motivations}\label{sec:issuesJeff}

As we mentioned in Section \ref{sec:refpriors}, we already know that the definitions of the reference prior are satisfied by the Jeffreys prior over the large class of priors admitting locally bounded and a.e. continuous densities w.r.t. the Lebesgue measure \citep{VanBiesbroeckBA2023}. 

This result is, however, limiting and disappointing in some cases. The reasons are the following ones: (i) the Jeffreys prior is not recommended in high-dimensional problems as it is known to be `either too diffuse or too concentrated' \citep{Berger2015}; moreover (ii) when the expression of the likelihood is itself complex, the computation of the Jeffreys prior can become  intractable, which is why (iii) in practice a restriction to the class of priors to ones which are easier to compute is often favored; also (iv) the Jeffreys prior is known to often lead to an improper prior, which does not necessarily issue a proper posterior distribution, essential for practical \emph{a posteriori} inference and sampling.

The first item mentioned in the above paragraph is frequently tackled with a sequential construction of the reference prior as suggested by \citet{Bernardo1979}.
On the condition that an ordering of the parameters is set:
 \begin{equation}
     \theta = (\theta_1,\dots,\theta_r) \in \Theta=\Theta_1\times\dots\times\Theta_r,
 \end{equation}
this construction considers a hierarchical construction of the reference prior. Typically, it is recommended to assume $\Theta_j\subset\RR^{d_j}$ with small dimensions $d_j$ (e.g., lower or equal than $2$) for any $j\in\{1,\dots,r\}$, and to sequentially build a reference prior on the $\Theta_j$,  $j\in\{1,\dots,r\}$:
 \begin{enumerate}
     \item initially fix $\ell_k^1=\ell_k$;
     \item for any values of $\theta_{j+1},\dots,\theta_r\in\Theta_{j+1}\times\dots\times\Theta_r$, compute a reference prior $\pi_j(\cdot|\theta_{j+1},\dots,\theta_r)$ under the model with likelihood $\theta_j\mapsto\ell_k^j(\mbf y|\theta_j,\dots,\theta_r)$;
     \item derive $\ell_k^{j+1}$ such as 
         \begin{equation}\label{eq:hier:condlikeint}
            \ell_k^{j+1}(\mbf y|\theta_{j+1},\dots,\theta_r) =  \int_{\Theta_j}\ell_k^j(\mbf y|\theta_j,\dots,\theta_r)d\pi_j(\theta_j|\theta_{j+1},\dots,\theta_r).
         \end{equation}
 \end{enumerate}
Step 2 of the method depicted above consists of the derivation of a reference prior w.r.t. the variable $\theta_j$ in the sense of Definition \ref{defi:refprior} (or of a quasi-reference prior if that latter does not exist). Thus, our (quasi)-reference priors over constrained classes of priors can be  plainly incorporated into this method.
Moreover, we invite to note that this construction does not solve the limitations (ii) to (iv) previously evoked. Actually, it makes them essential. Indeed, step 2 requires, firstly, a derivation of a reference prior, so that it would lead to a low-dimensional Jeffreys prior if the class of priors is not constrained. Also step 3 necessitates, secondly, that the latter leads to a proper posterior so that the integral involved does not diverge.

This last issue is taken into account by \citet{Berger1992} with the suggestion of such construction on an increasing sequence of compact subsets of $\Theta$: $\bigcup_{i\in\NN}\Theta_i=\Theta$. The hierarchical reference prior can then be chosen as a limit of the ones obtained under $\Theta_i$ when $i\to\infty$. However, this limit can be cumbersome to derive in practice. Another solution suggested by \citet{Mure2018} is to restrict the $\sigma$-algebra $\sY$ until the reference prior derived in step 2 leads to a proper posterior. It is still imperfect, as there is no guarantee that such a restricted $\sigma$-algebra exists outside the trivial one.

In the following subsections, we propose a range of solutions to some of the issues aforementioned, based on the derivation of reference priors over constrained classes of priors. In  Section \ref{sec:constrainedasympt}, we derive a quasi $D_\alpha$-reference prior over classes of priors that are easy to compute, in order to tackle the limitations (ii) and (iii) previously evoked. Then, Section \ref{sec:constrainedproper} explores another kind of constrained classes of priors, which leads to $D_\alpha$-reference priors that can solve the item (iv).

\subsection{Constrained reference priors based on Jeffreys' asymptotics}\label{sec:constrainedasympt}

In this section, we tackle the computational cost of the reference prior. As mentioned in Section \ref{sec:issuesJeff}, the Jeffreys prior expression is often complex to derive even in low dimensional models. This is even more a problem in practical studies where the prior must be evaluated a numerous number of times, when it resorts to MCMC simulations to provide posterior samples of $\theta$ for instance.

In some works of the literature, the Jeffreys prior is replaced by its decay rates at the boundary of the domain. For instance, the reference prior for Gaussian processes suggested by \citet{Gu2019} is built on the basis of the decay rates of a Jeffreys prior sequentially computed on the different variables that compose $\theta$ (following the construction presented in Section \ref{sec:issuesJeff}).
Their idea is that, in particular when it is improper, the prior provides the most information from its asymptotic rates, and variations of them are noticed to have a strong influence on the posterior distribution.
The result that follows provides a formalization of this intuition, focusing on the case where the Jeffreys prior asymptotically behaves like exponentiations of coordinates of $\theta$.

\begin{thm}\label{thm:Jthetaa}
    Suppose $\Theta\subset\RR$ is an interval of the form $[c,b)$ (or $(b,c]$). %
    \begin{itemize}
        \item If $b\in\RR$ and $J(\theta)\equi{\theta\rightarrow b}C|\theta-b|^a$ for constants $C\in\RR$ and $a\leq-1$, then $\pi^\ast(\theta)\propto|\theta-b|^a$ is the quasi $D_\alpha$-reference prior over $\hat\cP = \{\pi(\theta)\propto|\theta-b|^u,\,u\in\RR\}$.
        \item If $|b|=\infty$ and $J(\theta)\equi{\theta\rightarrow b}C\theta^a$ for constants $C\in\RR$ and $a\geq-1$, then $\pi^\ast(\theta)\propto\theta^a$ is the quasi $D_\alpha$-reference prior over $\hat\cP = \{\pi(\theta)\propto\theta^u,\,u\in\RR\}$.
    \end{itemize}
\end{thm}

\begin{proof}
    The proof is technical and detailed in section \ref{sec:proofs}.
    The idea is that $l(\pi)$ can be seen as a negative divergence between $\pi$ and $J$. However, when $J$ is improper at the boundary of the domain, the maximization of $l(\pi)$ gets closer to the minimization of a divergence between $\pi$ and the improper decay rate of $J$.
\end{proof}

\begin{rem}
    Theorem \ref{thm:Jthetaa} still stands when $\Theta=(b,c)$ (or $(c,b)$) if $c\ne\infty$ and if $J(\theta)$ admits a non-null and finite limit when $\theta\to b$.
\end{rem}

This theorem serves the statement of two conclusions: (i) it emphasizes that when Jeffreys prior is improper, its improper decay rates contain the most relevant information, and (ii) it proposes to choose this asymptotic expansion of Jeffreys as a quasi $D_\alpha$-reference prior when we look for an easy prior to compute.

\subsection{Constraints for proper reference posteriors}\label{sec:constrainedproper}

In Section \ref{sec:constrainedasympt}, we have provided some elements to construct a tractable reference prior on the coordinates over which Jeffreys prior is improper.
The reference prior that our theorem proposes keeps the improper characteristic of Jeffreys prior on the same coordinates.
This improper aspect can, however, remain an issue in some cases, especially when the resulting posterior is improper as well.

For this reason, it might happen that some asymptotic rates in some directions still have to be tackled. The work in this section is concluded by results that allow defining a $D_\alpha$-reference prior (or quasi $D_\alpha$-reference prior), which benefits from adjusted decay rates from Jeffreys prior. 
The proposition below constitutes a preliminary result that gives the form of a $D_\alpha$-reference prior over a class of priors with linear constraints.

\begin{assu}\label{assu:glibre}
    A family of functions from $\Theta$ to $\RR$ $(g_j)_{j=1}^p$ is said to satisfy Assumption \ref{assu:glibre} if $g_0,\dots,g_p$ are linearly independent in the space of a.e. continuous functions from $\Theta$ to $\RR$, where $g_0:=\theta\mapsto 1$.
\end{assu}

\begin{prop}\label{prop:constraints}
    Suppose $\Theta$ to be a compact subset of $\RR^d$. Let $g_1,\dots,g_p$ be %
    functions in $\cR^\cC$ %
    that satisfy Assumption \ref{assu:glibre}. Define $\tilde\cP$ the class of priors on $\Theta$ such that $\forall 1,\dots,p$, $\int_\Theta \pi g_j=c_j$, for some $c_j\in\RR$.
    If %
    there exists a $D_\alpha$-reference prior over $\tilde\cP$, it is unique. If it is a.e. positive, its density $\pi^\ast$ verifies
    \begin{equation}
        \pi^\ast(\theta) = J(\theta)\left(\lambda_0+\sum_{j=1}^p\lambda_jg_j(\theta) \right)^{1/\alpha},
    \end{equation}
    for some $\lambda_j\in\RR$. Reciprocally, if there exists a prior $\pi^\ast\in\tilde\cP$ which verifies the above equation for some $\lambda_j\in\RR$, it is the $D_\alpha$-reference prior over $\tilde\cP$.
\end{prop}

\begin{proof}
    This proposition results from a Lagrange multipliers theorem. A detailed proof is proposed in Section \ref{sec:proofs}.
\end{proof}

\begin{rem}
    While it is not the subject of this work, we let the reader notice that this proposition opens the way to the introduction of constraints based on expert judgments in prior elicitation. They can take the form of moment constraints or predictive constraints \citep{Bousquet2023}.    
\end{rem}

\begin{rem}\label{rem:klconst}
    The expression of the reference prior given by Proposition \ref{prop:constraints} depends on the chosen $\alpha$-divergence.
    While this work considers only the framework of reference priors under $\alpha$-divergences as a dissimilarity measure, a version of this theorem could be written in the original framework of the reference prior theory that uses the Kullback-Leibler divergence. The expression of the resulting reference prior would be impacted. In \citep{BernardoSmith1994}, \citeauthor{BernardoSmith1994} suggest that the expression should take the form of
        \begin{equation}
            \pi^\ast\propto J\cdot\exp\left(\sum_{j=1}^p\lambda_j g_j\right),
        \end{equation}
        for some $\lambda_j$ that remain to be determined.
\end{rem}

Below, given a function $g$ that is selected to adjust the asymptotics of Jeffreys prior, is stated the expression of a proper $D_\alpha$-reference prior.

\begin{thm}\label{thm:lintoproper}
    Let $g:\Theta\to(0,\infty)$ be a function in $\cR^\cC$ such that
        \begin{equation}\label{eq:intgalphafinite}
            \int_\Theta J(\theta)g^{1/\alpha}(\theta)d\theta<\infty \quad\text{and}
            \quad\int_\Theta J(\theta)g^{1/\alpha+1}(\theta)d\theta<\infty,
        \end{equation}
    and suppose that $g$ is bounded in the neighborhood of $b$ for an element $b\in\partial\Theta$. %
    We denote by $\overline\cP$ the class of positive priors $\pi$ on $\Theta$ such that $\int_\Theta\pi g<\infty$, and we define 
    $\pi^\ast\in\overline\cP$ as follows 
        \begin{equation}
            \pi^\ast(\theta)\propto J(\theta)g(\theta)^{1/\alpha}.
        \end{equation}
    If $Jg$ is non-integrable in the neighborhood of $b$, then $\pi^\ast$ is a $D_\alpha$-reference prior over $\overline\cP$. Otherwise, and if $J$ is improper in the neighborhood of $b$, $\pi^\ast$ is a $D_\alpha$-reference prior over the class of proper priors in $\overline\cP$.
\end{thm}

\begin{proof}
    The statement of this theorem results from the sequential use of Proposition \ref{prop:constraints} on an increasing sequence of compact subsets of $\Theta$. A detailed proof is written in Section \ref{sec:proofs}.     
\end{proof}

   To improve the above theorem, one would like 
     to relax the first assumption in Equation (\ref{eq:intgalphafinite}), i.e., to let $\int_\Theta J g^{1/\alpha}$  be infinite. Indeed, in this way, the result would provide a reference prior $\pi^\ast$ ---non-necessarily proper--- but such that $\pi^\ast g\in L^1$. With a good choice of $g$, $\pi^\ast$ could be built as a prior that provides a proper posterior. It is the purpose of the next theorem. 
    The cost of this relaxation is the provision of a quasi-reference prior instead of a reference prior.

\begin{thm}\label{thm:quasipostpropre}
    Let $g:\Theta\to(0,\infty)$ be a continuous function such that
    \begin{equation}
        \int_\Theta J(\theta)g(\theta)d\theta=\infty\quad\text{and}\quad \int_\Theta J(\theta)g^{1/\alpha+1}(\theta)d\theta<\infty,
    \end{equation}
and suppose that $g(\theta)\conv{\theta\rightarrow b}0$ for an element $b\in\partial\Theta$ such that $J$ is non-integrable in the neighborhood of $b$.\\
Let $(\Theta_i)_{i\in\NN}$ be an openly increasing sequence of compact sets that covers $\Theta$ and $(c_i)_i$ be a bounded cequence in $(0,\infty)$. Define the class of priors  $\overline\cP'=\{\pi,\,\forall i,\,\int_{\Theta_i}\pi g=c_i\int_{\Theta_i}\pi\}$.\\ %
Denote for any $i$ $\overline{\cP}'_i$ the class of renormalized restrictions to $\Theta_i$ of priors in $\overline{\cP}'$. If for any $i$ there exists a  positive maximum of $l$ over $\overline\cP'_i$, then  $\pi^\ast$ is a quasi $D_\alpha$-reference prior over $\overline\cP'$ with
    \begin{equation}
        \pi^\ast(\theta)\propto J(\theta)g(\theta)^{1/\alpha}.
    \end{equation}
    This prior is such that $\int_\Theta\pi^\ast g<\infty$.
\end{thm}

\section{Discussions}\label{sec:discussion}

The knowledge of Jeffreys prior's decay rates is central in the results presented in this work. These results indicate that in common scenarios where Jeffreys prior is improper, these rates must be explicitly considered in order to construct a reference prior. 

We let the reader note that Theorems \ref{thm:lintoproper} and \ref{thm:quasipostpropre} introduce results that also depend on the chosen dissimilarity measure. Therefore, a balance must be found between the subjective influence of the constraint and the quest for an informed prior to facilitate possible sampling from the posterior. This is illustrated in the example we address in the following section.
However, it is important to observe that using the KL-divergence instead of the $\alpha$-divergence %
would result in a stronger influence of the constraint on the final prior. As noted in Remark \ref{rem:klconst}, the exponentialization of the function $g$ could lead to a prior with distribution tails that are significantly negligible beyond those of Jeffreys, thereby jeopardizing its objective nature.

Generally, the results we propose in Sections \ref{sec:constrainedasympt} and \ref{sec:constrainedproper} address different problems and are thus  fundamentally different in nature. In one case, the improper aspect of Jeffreys prior is not necessarily challenged, and an efficient construction of the latter is proposed. In the other case, the goal is to significantly attenuate its improper aspect while maintaining as much objectivity as possible. In this latter case, however, the expressions of the proposed  reference priors still depend on the expression of Jeffreys prior. Nevertheless, when Jeffreys prior is proper, there is no guarantee that a straightforward construction inspired by its convergence rates at the domain boundaries will be relevant. %
Indeed, although improper tails concentrate an infinite mass that constitutes all the information at the boundaries, when they are proper, the information of interest may need to be sought elsewhere. In this case, a calculation or approximation of the `proper' Jeffreys prior remains to be considered.

Finally, regarding Theorem \ref{thm:Jthetaa}, although the result is limited to parameter power distribution tails, it is observed that, in practice, these include a wide range of improper Jeffreys priors.
For example, this includes Jeffreys priors derived from various Gaussian models, such as those introduced by \citet{Neyman1948}; Jeffreys priors related to specific parameters within Gaussian process models \cite{Gu2016}; and those arising in more specialized contexts, like the one in \cite{VanBiesbroeck2023}.
Moreover, the invariance of Jeffreys priors under reparameterization can sometimes allows us to return to this case. Specifically, if $J$ can be asymptotically written as a power of a function $f$, where $f$ is differentiable, monotone and with bounded derivative (from above and from below), then the reparameterization $\vartheta=f(\theta)$ %
should allow us to recover the reference prior among those expressible as powers of $f$.

In the following section, we illustrate an application of our work with an example taken from the literature.

\section{An example}\label{sec:example}

In their work, \citet{Rubio2014} prove that the two piece location-scale model they proposed has an improper Jeffreys prior, which issues an improper posterior.
The model is parameterized by $\theta=(\mu,\sigma_1,\sigma_2)\in\RR\times(0,\infty)^2$, inferred over observations in $\cY=\RR$. It has the following likelihood:
    \begin{equation}\label{eq:examplelikilihood}
        \ell(y|\theta) =  \frac{2}{\sigma_1+\sigma_2}\left[f\left(\frac{y-\mu}{\sigma_1}\right)\indic_{(-\infty,\mu)}(y) + f\left(\frac{y-\mu}{\sigma_2}\right)\indic_{(\mu,\infty)}(y)  \right],
    \end{equation}
where $f$ is a density function with support on $\RR$, assumed to be symmetric with a single mode at zero, and with a few integrability assumptions that are detailed in \citep{Rubio2014}.
    The choice of $f$ is open and can let the assumptions of Theorem \ref{thm:l(pi)} to be verified.
    We may take, for instance, the standard Gaussian density function.
    Under this construction, the Fisher information matrix of this model takes the form:
    \begin{equation}\label{eq:fishermatrix}
        \cI(\theta) = \left(\begin{array}{ccc}
             \frac{\alpha_1}{\sigma_1\sigma_2}& -\frac{2\alpha_3}{\sigma_1(\sigma_1+\sigma_2)} & \frac{2\alpha_3}{\sigma_2(\sigma_1+\sigma_2)}  \\
             \ast &\frac{\alpha_2}{\sigma_1(\sigma_1+\sigma_2)} + \frac{\sigma_2}{\sigma_1(\sigma_1+\sigma_2)^2} & - \frac{1}{(\sigma_1+\sigma_2)^2}  \\
             \ast&\ast& \frac{\alpha_2}{\sigma_2(\sigma_1+\sigma_2)} + \frac{\sigma_1}{\sigma_2(\sigma_1+\sigma_2)^2}
        \end{array}\right)
    \end{equation}
    for some positive constants $\alpha_1$, $\alpha_2$ and $\alpha_3$.

    The full Jeffreys prior can be computed as 
    \begin{equation}
        J(\theta) \propto \frac{1}{\sigma_1\sigma_2(\sigma_1+\sigma_2)},
    \end{equation}
    it is improper and leads to an improper posterior. In the following, we construct different priors based on the suggestions developed in this paper. %

    \paragraph{Proper priors based on a moment constraint}
    Considering results in Section \ref{sec:constrainedproper} and the decay rates of the Jeffreys prior above, a simple correction can be done to issue a proper reference prior w.r.t. $\sigma_1,\sigma_2$, which results in a proper posterior.
    We consider $\alpha\in(0,1)$; given $\eps\in(0,\frac{1}{1+1/\alpha})$ we have 
        \begin{equation}
            \int J(\theta)(\sigma_1\sigma_2)^{\eps/\alpha+1} d\sigma_1d\sigma_2<\infty,
        \end{equation}
    so that the associate proper $D_\alpha$-reference prior $\pi^\ast$, which is such that $\pi^\ast(\mu,\sigma_1,\sigma_2)\sigma_1^\eps\sigma_2^\eps$ is integrable w.r.t. $\sigma_1,\sigma_2$, is
        \begin{equation}
            \pi^\ast(\mu,\sigma_1,\sigma_2) \propto \frac{(\sigma_1\sigma_2)^{\eps/\alpha -1}}{\sigma_1+\sigma_2}.        
        \end{equation}

\paragraph{Overview and sensibility on the parameters}
Our prior densities are compared with the Jeffreys prior on Figure \ref{fig:priorpost}.(a), w.r.t. the parameter $\sigma_2$ the others being fixed to $1$. For this comparison, a multiplicative constant had to be chosen on $J$, we have chosen the one such that $J(1,1,1)=2$.
Our priors differ as a function of $\gamma=\eps/\alpha\in(0,\frac{1}{1+\alpha})\subset(0,1)$. On the one hand, when $\gamma$ becomes close to $0$, the prior ---which we denote by $\pi_\gamma^\ast$ from now on--- becomes close to the Jeffreys prior, i.e., the most objective prior w.r.t. the mutual information criterion. However, in this case, $\pi^\ast_\gamma$ becomes close to an improper prior, and its posterior becomes close to an improper posterior. On the other hand, setting $\gamma$ away from $0$ 
rearranges the quantity of information in the prior. Its referential nature decreases in favor of an increase in its entropy.
Therefore, a trade-off has to be made between suitability for inference and objectivity.
Finally, note that in Figure \ref{fig:priorpost}.(a) is also drawn the \emph{independent Jeffreys prior} proposed by the authors in \citep{Rubio2014} for this model as an alternative to Jeffreys. Its decay rates are actually the same as the ones of $\pi^\ast_\gamma$ when $\gamma=1/2$, and the two priors are hard to distinguish. Note that this latter prior $\pi^\ast_{1/2}$ equals the hierarchical reference prior (as described in Section \ref{sec:issuesJeff}) constructed from the ordering $\pi(\theta)=\pi_1(\sigma_1,\sigma_2|\mu)\pi_2(\mu)$.

To evaluate a bit further our method, we propose a visualization of the posterior sensitivity to the priors, i.e., to $\gamma$. 
Such influence quantification constitutes a critical step of the \emph{Bayesian workflow}, as expressed by   \citet{gelman2020bayesianworkflow}. Several methods exist  in the literature for this purpose (e.g., \citep{Berger1990, Nott2020}). 
An approach is to compare the variations of an \emph{a posteriori} quantity as a function of the parameter \citep{Kallioinen2023}. In this example, this methodology is yet limited by the improper aspect of Jeffreys posterior. It cannot be considered for any comparison.
In Figure \ref{fig:priorpost}.(b), (c) and (d) are plotted, for numerous $\gamma$ and several data set sizes $k$, the posterior densities that result from a sample of data and from the prior $\pi^\ast_\gamma$.
As expected, the influence of the prior appears for small values of $k$.
Indeed, we can notice that for higher values of $\gamma$, the posterior is slightly more shifted to the right and seems to be a little flatter. 
It is remarkable that this observation becomes limited when $k$ increases. When $k=50$, the difference between the posterior densities 
is hard to distinguish. %
This indicates that the little losses of objectivity should induce small variations in the resulting inference in this example.

For practical details, the chosen $f$ is the standard Gaussian density, and the different data sets have been generated according to the likelihood in Equation (\ref{eq:examplelikilihood}) conditionally to $\theta^\ast=(2,2,2)$. %

\begin{figure}[!ht]%
    \centering%
    \includegraphics[width=0.45\linewidth]{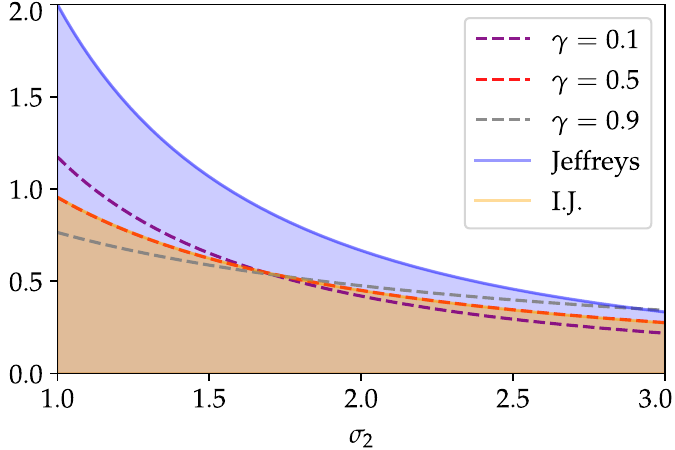}%
    \includegraphics[width=0.45\linewidth]{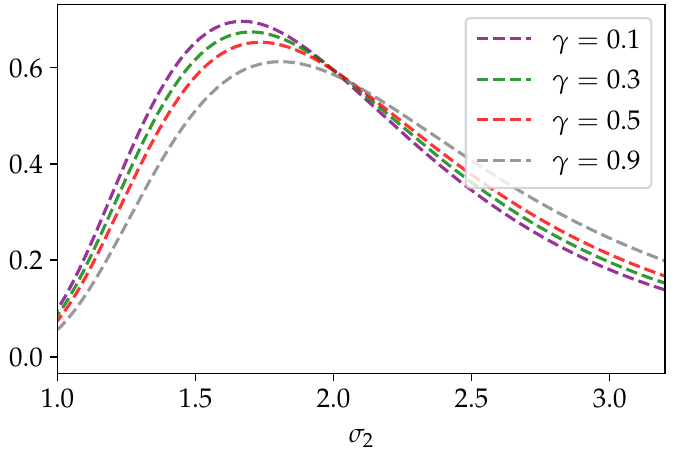}\\
    {~\hspace{\stretch{1}}(a)\hspace{\stretch{2}}(b)\hspace{\stretch{1}}~}\\[2pt]
    \includegraphics[width=0.45\linewidth]{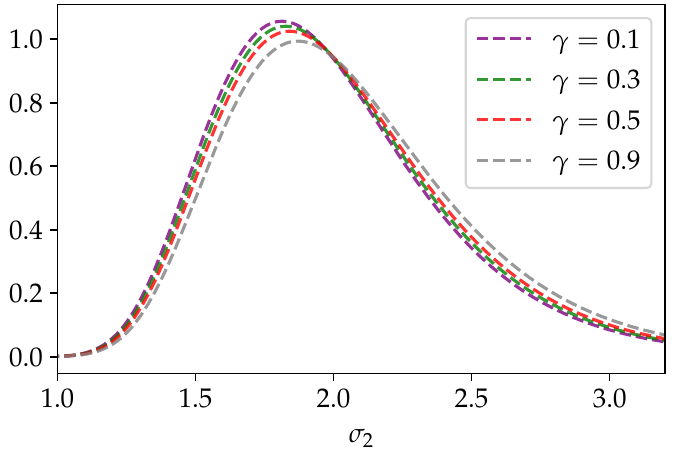}%
    \includegraphics[width=0.45\linewidth]{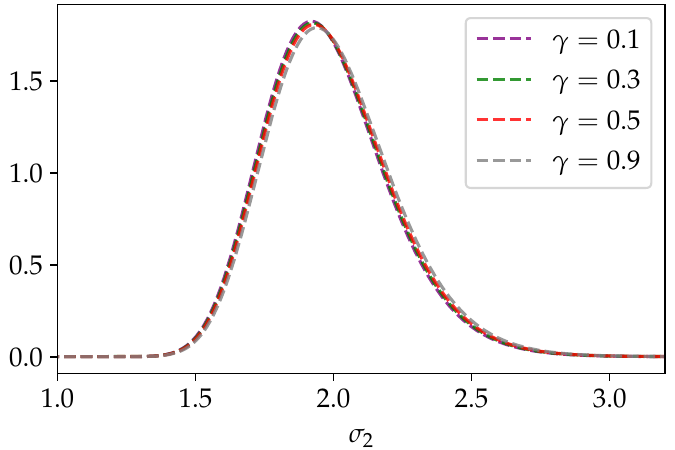}\\
    {~\hspace{\stretch{1}}(c)\hspace{\stretch{2}}(d)\hspace{\stretch{1}}~}%
    \caption{On (a), different prior densities $\pi^\ast_\gamma$ (in dashed line) w.r.t. $\sigma_2$ along with a Jeffreys prior density (that delimits the blue area) and the \emph{Independent Jeffreys} of \citep{Rubio2014} (that delimits the orange area). On each of (b), (c) and (d), posterior densities for different values of $\gamma$ are plotted. %
    One observed sample was used to calculate all the posteriors in a figure.
    The observed sample has a size of $k=5$ in (b), $k=15$ in (c), and $k=50$ in (d).\label{fig:priorpost}}
\end{figure}

\section{Detailed proofs}\label{sec:proofs}

\subsection{Proof of Proposition \ref{prop:kolmog}}

Let $1\leq\mbf i\leq\infty$, $I=\{i\in\NN,\,i<\mbf i\}$, and $(U_i)_{i\in I}$ be a sequence of disjoint non-empty sets in $\sT$ that covers $\Theta$ such that $\pi(U_i)<\infty$ $\forall i$.
When $\pi(U_i)=0$, we denote by $\pi(\cdot|U_i)$ an arbitrary probability distribution on $U_i$ (for instance a Dirac distribution); otherwise, it is defined by $\pi(\cdot|U_i)=\pi(\cdot\cap U_i)/\pi(U_i)$.
    Fix firstly $i\in I$ and for any $k\geq1$ let us construct $\PP_{\mbf Y_k,T}^i$ on $\sY^{\otimes k}\otimes\sT$ such that
        \begin{equation}
            \PP_{\mbf Y_k,T}^i(A_1\times\dots\times A_k\times B) = \int_B \PP_{Y|\theta}^{\otimes k}(A_1\times\dots\times A_k)d\pi(\theta|U_i),
        \end{equation}
     so that $\PP^i_{\mbf Y_k,T}$ is a probability distribution on $\cY^k\times U_i$. 
    Therefore, using Kolmogorov extension theorem, there exist a probability space $(\Omega^i,\sX^i,\PP^i)$, a random process $\overline{\mbf Y^i}=(Y^i_j)_{j\geq1}$, $Y_j^i:\Omega^i\to\cY$, and a random variable $T^i:\Omega^i\to\Theta$ such as  their joint distribution is uniquely defined by the distribution
    $\PP^i_{\mbf Y_k,T}$ of $(T^i,\mbf Y_k^i)$ for any $k$, with $\mbf Y_k^i=(Y_j^i)_{j=1}^k$.

Denote $\mbf\Omega= I \times\prod_{i\in I}\Omega^i$. A natural $\sigma$-algebra on $\mbf\Omega$ is the one generated by the cylinder sets:
\begin{equation}
    \mbf\Xi = \sigma\big(N\times E_0\times\dots\times E_{k-1}\times\prod_{k\leq i<\mbf i}\Omega^i,\,N\in\sP(I),\, E_0,\dots,E_k\in\sX^i,\,k<\mbf i \big).
\end{equation}
We can define $\Pi$ by
    \begin{equation}
        \Pi\Big(N\times\prod_{i\in I}E_i\Big) = \sum_{n\in N}\PP^n(E_n)\pi(U_n),
    \end{equation}
for any $N\in\sP(I)$, $(E_i)_{i \in I} \in\prod_{i\in I}\sX^i$ with $E_i=\Omega^i$ for any $i\geq k$ for some $k<\mbf i$. 
The existence of the measure $\mbf\Pi$ on $\mbf\Xi$, that coincides with $\Pi$ on the cylinder sets is guaranteed by  Carathéodory's extension Theorem. %

Now, we call $p^n:(m,(\omega^i)_i)\in\mbf\Omega\mapsto \omega^n$ for any $n\in I$, $c:(m,(\omega^i)_i)\in\mbf\Omega\mapsto m$, and $p:\mbf w\in\mbf\Omega\mapsto p^{c(\mbf w)}(\mbf w)$.
This way, we can define $T:\mbf\Omega\to\Theta$ by $T(\mbf w)=T^{c(\mbf w)}(p(\mbf w))$ and $\overline{\mbf Y}=(Y_j)_{j\geq1}$ with for any $j\geq1$, $Y_j:\mbf\Omega\to\sY$ by $Y_j(\mbf w)=Y_j^{c(\mbf w)}(p(\mbf w))$.\\
We can verify that $T$ is a measurable: if $B\in\sT$,  
    \begin{equation}
        T^{-1}(B) = \bigcup_{i\in I}\{\mbf w\in\mbf\Omega,\, c(\mbf w)=i,\,T^i(p(\mbf w))\in B  \} = \bigcup_{i\in I}c^{-1}(\{i\})\cap (T^i(p^i))^{-1}(B)
    \end{equation}
which is measurable. 
The same arguments stand for the measurability of $Y_j$ for any $j$.

We have,
for any $V\in\sT$:
    \begin{align}
          \mbf\Pi(T\in V) &= \Pi\Big(\bigcup_{i\in I}\{i\}\times\Omega^1\times\dots\times\{T^i\in V\}\times\Omega^{i+1}\times\dots\Big)\\
          &=\sum_{i\in I}\PP^i(T^i\in V)\pi(U_i) = \sum_{i\in I}\pi(V|U_i)\pi(U_i) = \pi(V);\nonumber
    \end{align}
and if $\pi(V)<\infty$ then for any $(A_{j_l})_{l=1}^k\in\sY^k$: %
    \begin{multline}
        \mbf\Pi\left((Y_{j_l})_{l=1}^k\in(A_{j_l})_{l=1}^k|T\in V\right)\mbf\Pi(T\in V) \\ = 
        \mbf\Pi\Big(\bigcup_{i\in I}  \{i\}\times\Omega^1\times\dots\{(Y^i_{jl})_{l=1}^k\in(A_{jl})_{l=1}^k\}\cap \{T^i\in V \}\times\Omega^{i+1}\dots 
        \Big),
    \end{multline}
    so that
        \begin{align}\nonumber
        &\mbf\Pi\left((Y_{j_l})_{l=1}^k\in(A_{j_l})_{l=1}^k|T\in V\right) 
            = \sum_{i\in I}\PP^i\left( \{(Y_{j_l}^i)_{l=1}^k\in(A_{j_l})_{l=1}^k\}\cap \{T^i\in V\}\right)\frac{\pi(U_i)}{\pi(V)} \\
            &= \sum_{i\in I}\int_{V\cap U_i}\PP^{\otimes k}_{Y|\theta}\left((A_{j_l}\right)_{l=1}^k)\frac{d\pi(\theta)}{\pi(V)}  = \int_\Theta\PP^{\otimes k}_{Y|\theta}\left((A_{j_l})_{l=1}^k\right)d\pi(\theta|V).
    \end{align}

We precise that the above expressions rely on the fact that $\mbf\Pi(I\times\prod_{i\in I}E_i)=\sum_{i\in I}\PP^i(E_i)\pi(U_i)$, for any $E_i\in\sX^i$, for every ${i\in I}$. To state that equality, we write $I\times\prod_{i\in I}E_i=\bigcup_{n\geq1}\bigcap_{k<\mbf i}(I^n\times\prod_{i\in I}E_{i,k}') $ with $I^n=\{i\in I,\,i< n\}$ and $E_{i,k}'=E_i$ if $i\leq k$, $E_{i,k}'=\Omega^i$ otherwise. This way for any $n\geq1$:
$I^n\times\prod_{i\in I} E_{i,k}' \subset I^{n+1}\times\prod_{i\in I} E_{i,k}'$ for any $k<\mbf i$, and
$        \bigcap_{k<\mbf i} I^n\times\prod_{i\in I} E_{i,k}'\subset \bigcap_{k<\mbf i} I^{n+1}\times\prod_{i\in I} E_{i,k}'$.
Thus,
    \begin{equation}\label{eq:proofKolmLimPiCup}
        \mbf\Pi\Big(I\times\prod_{i\in I}E_i\Big) 
            = \lim_{n\rightarrow\infty}\mbf\Pi\Big(\bigcap_{k<\mbf i}(I^n\times\prod_{i\in I}E_{i,k}')\Big).
    \end{equation}

For any $n\geq1$, $ %
(I^n\times\prod_{i\in I}E_{i,k}')$ is a decreasing sequence, because $E_{i,k+1}'\subset E_{i,k}'$ for any $k<\mbf i-1$. 
Also $\mbf\Pi(I^n\times\prod_{i\in I}E_{i,0}')=\sum_{i=0}^{n-1}\pi(U_i)<\infty$.
Therefore, we can write:
    \begin{align}
        \mbf\Pi\Big(\bigcap_{k<\mbf i}(I^n\times\prod_{i\in I}E_{i,k}')\Big) 
            &= \lim_{k\rightarrow\infty} \mbf\Pi\Big(I^n\times\prod_{i<k,\mbf i}E_i\times\prod_{k\leq i<\mbf i}\Omega^i\Big) \\
            &= \lim_{k\rightarrow\infty} \Big[\sum_{i<k,n,\mbf i}\PP^i(E_i)\pi(U_i)  + \sum_{k\leq i<n,\mbf i}\pi(U_i)\Big]=\sum_{i<n,\mbf i}\PP^i(E_i)\pi(U_i),\nonumber
    \end{align}
because the sums are all finite. Eventually, the limit in Equation (\ref{eq:proofKolmLimPiCup}) is equal to $\sum_{i<\mbf i}\PP^i(E_i)\pi(U_i)$ as expected, hence the result.

\subsection{Proof of Theorem \ref{thm:Jthetaa}}

To prove this theorem, we consider to simplify the derivations that $b=0$ with $\Theta=(0,1]$. We will show later how to extend the result to the other cases. As the Jeffreys prior can be defined up to a positive multiplicative constant with no incidence on the definition of a $D_\alpha$-reference prior, 
we will simplify its decay rate assuming $J(\theta)\equi{\theta\rightarrow 0}\theta^a$.

Let us define the increasing sequence of compact subsets $\Theta$: $\Theta_i=[\theta_i,1]$, $i\geq0$, with $\theta_i\conv{i\rightarrow\infty}0$. We denote by $\psi_i$ and $\tilde\psi_i$ the functions defined as follow:
    \begin{equation}
        \psi_i(u) = -\frac{\int_{\Theta_i}J(\theta)^{-\alpha}\theta^{u(1+\alpha)}d\theta }{\left(\int_{\Theta_i}\theta^ud\theta \right)^{1+\alpha}},\qquad \tilde\psi_i(u) = -\frac{\int_{\Theta_i}\theta^{-a\alpha }\theta^{u(1+\alpha)}d\theta }{\left(\int_{\Theta_i}\theta^ud\theta \right)^{1+\alpha}}.
    \end{equation}
The quantity $\psi_i(u)$ corresponds ---up to a positive constant--- to the parametrization w.r.t. $u\in\RR$ of $l(\pi_i)$; where $\pi(\theta)\propto\theta^u$, $\pi_i$ is re-normalized restriction of $\pi$ to $\Theta_i$  and where $l$ is the function defined in Theorem \ref{thm:l(pi)} that we seek to maximize.
Therefore, if we call $u_i\in\argmax_{u\in\RR}\psi_i(u)$ for any $i$, and if the associated sequence of priors $(\pi_i)_i$ converge Q-vaguely to a prior $\pi^\ast$ over $\Theta$, that would prove that $\pi^\ast$ is a quasi $D_\alpha$-reference prior.

\paragraph{The case $\mathbf{a<-1}$}
    Firstly, we assume that $a<-1$. Denote $U=(-\infty, u_m:=\frac{\alpha a-1}{1+\alpha})$, we want to derive an asymptotic equivalent when $i\to\infty$ of $\psi_i(u)$, uniformly w.r.t. $u\in U$.
    More precisely, we are about to show that for any $\eps>0$ there exists a $i_1$ such that for any $u\in U$ and $i>i_i$, $|\psi_i(u)-\tilde\psi_i(u)|<\eps|\tilde\psi_i(u)|$. %

    Let $\eps>0$, using that $J(\theta)^{-\alpha}\equi{\theta\rightarrow0}\theta^{-\alpha a}$, there exists $i_0$ such that for any $\theta<\theta_{i_0}$, $|J(\theta)^{-\alpha}-\theta^{-\alpha a}|<\eps\theta^{-\alpha a}$. %
    Now, we consider an $i_1$ such that: %
        \begin{equation}
            \frac{\int_{\Theta_{i_0}}J(\theta)^{-\alpha}d\theta }{\int_{\theta_{i_1}/\theta_{i_0}}^1J(\theta_{i_0}\theta)^{-\alpha}\theta^{u_m(1+\alpha)} d\theta}<\theta_{i_0}\eps/2\quad \text{and}\quad
            \frac{\int_{\Theta_{i_0}}\theta^{-\alpha a}d\theta }{\int_{\theta_{i_1}/\theta_{i_0}}^1(\theta_{i_0}\theta)^{-\alpha a}\theta^{u_m(1+\alpha)} d\theta}<\theta_{i_0}\eps/2.
        \end{equation}
    Thus, for any $i>i_1$  and any $u\in U$:
        \begin{multline}
            \frac{\int_{\Theta_{i_0}} J(\theta)^{-\alpha}\theta^{u(1+\alpha)}  d\theta}{\int_{\Theta_{i}\setminus\Theta_{i_0} } J(\theta)^{-\alpha }\theta^{u(1+\alpha)}d\theta }
            = 
            \frac{\theta_{i_0}^{-1}\int_{\theta_{i_0}}^1 J(\theta)^{-\alpha}(\frac{\theta}{\theta_{i_0}})^{u(1+\alpha)}  d\theta}{\int_{\theta_{i}/\theta_{i_0}}^1 J(\theta_{i_0}\theta)^{-\alpha }\theta^{u(1+\alpha)} d\theta }\\
            < \frac{\theta_{i_0}^{-1}\int_{\Theta_{i_0} }J(\theta)^{-\alpha} d\theta}{\int_{\theta_{i_1}/\theta_{i_0} }^1 J(\theta_{i_0}\theta)^{-\alpha }\theta^{u_m(1+\alpha)}d\theta} < \eps/2
        \end{multline}
    and
        \begin{multline}
            \frac{\int_{\Theta_{i_0}} \theta^{-\alpha}\theta^{u(1+\alpha)}  d\theta}{\int_{\Theta_{i}\setminus\Theta_{i_0} } \theta^{-\alpha }\theta^{u(1+\alpha)}d\theta }
            = 
            \frac{\theta_{i_0}^{-1}\int_{\theta_{i_0}}^1 \theta^{-\alpha}(\frac{\theta}{\theta_{i_0}})^{u(1+\alpha)}  d\theta}{\int_{\theta_{i}/\theta_{i_0}}^1 (\theta_{i_0}\theta)^{-\alpha }\theta^{u(1+\alpha)} d\theta }\\
            < \frac{\theta_{i_0}^{-1}\int_{\Theta_{i_0} }\theta^{-\alpha} d\theta}{\int_{\theta_{i_1}/\theta_{i_0} }^1 (\theta_{i_0}\theta)^{-\alpha }\theta^{u_m(1+\alpha)}d\theta} < \eps/2,
        \end{multline}
    so that $|\psi_i(u)-\tilde\psi_i(u)|<\tilde\eps|\tilde\psi_i(u)|$ as expected. %

    Now we want to use this asymptotic equivalence to bound the difference $|a-u_i|$, where $u_i$ is defined as a maximal argument of $\psi_i$. The next step is thus to show that such $(u_i)_i$ exists.
    
    There exist $\tilde K$, $\tilde K'$ such that $\tilde K\theta^{-\alpha}\leq J(\theta)^{-\alpha}\leq \tilde K'\theta^{-\alpha}$. 
    Let $i\geq0$, we can write
        \begin{equation}\label{eq:gendarmeJeffreys}
            \tilde K |\tilde\psi_i(u)| \leq|\psi_i(u)| \leq \tilde K'|\tilde\psi_i(u)|
        \end{equation}
    with 
        \begin{equation}
            |\tilde\psi_i(u)| = \theta_i^{-\alpha(a+1)}\frac{|u+1|^{1+\alpha}}{|\alpha(u-a)+u+1|}\frac{1-\theta_i^{-\alpha(u-a)-u-1}}{(1-\theta_i^{-u-1})^{1+\alpha}}  \conv{u\rightarrow-\infty}+\infty.
        \end{equation}
    That makes $|\psi_i|$ being a coercive and continuous function on $U$, so that it admits minimal arguments in $U$. We denote by $u_i$ one of them: 
        \begin{equation}
            u_i \in \argmax_{u\in U} \psi_i(u).
        \end{equation}

    We remind that given \cite[Proposition 3]{VanBiesbroeckBA2023}, $a$ is the only maximal argument of $\tilde\psi_i$ for any $i$. 
    This way, for $i>i_1$, we write
        \begin{align}
            |\tilde\psi_i(u_i)-\tilde\psi_i(a)|&\leq |\psi_i(u_i)-\tilde\psi_i(u_i)| +  |\psi_i(a)-\tilde\psi_i(a)| + |\psi_i(u_i)-\psi_i(a)| \nonumber\\
            \tilde\psi_i(a) - \tilde\psi_i(u_i) &%
            \leq \eps(|\tilde\psi_i(u_i)|+|\tilde\psi_i(a)|)+ \psi_i(u_i)-\psi_i(a),
        \end{align}
    which leads to 
        \begin{align}
            2(\tilde\psi_i(a)-\tilde\psi_i(u_i)) &%
            \leq \eps(|\tilde\psi_i(u_i)|+|\tilde\psi_i(a)|) +  \psi_i(u_i)-\tilde\psi_i(u_i)+\tilde\psi_i(a)-\psi_i(a) \nonumber\\
            \tilde\psi_i(a) - \tilde\psi_i(u_i) &%
            \leq \eps(|\tilde\psi_i(u_i)|+|\tilde\psi_i(a)|).
        \end{align}
    Consequently to the convergence of
     $(\theta_i^{\alpha(a+1)}\tilde\psi_i(a))_i$ toward a positive limit when $i\to\infty$, we deduce that $\theta_i^{\alpha(a+1)}(\tilde\psi_i(a)-\tilde\psi_i(u_i))/|\theta_i^{\alpha(a+1)}\tilde\psi_i(u_i)|$ is asymptotically null. This prevents the sequence $(\theta_i^{\alpha(a+1)}\tilde\psi_i(u_i))_i$ to admit a non finite subsequential limit, meaning it has to be bounded and to converges to the same limit as $(\theta_i^{\alpha(a+1)}\tilde\psi_i(a))_i$, i.e. $-|a+1|^\alpha$.

     On another hand, we notice that for any $M>0$, there exist a $M'$ such that for any $u<M'$, 
        \begin{equation}
            \frac{|u+1|^{1+\alpha}}{|\alpha(u-a)+u+1|}>M
        \end{equation}
     and $|\theta^{\alpha(a+1)}\tilde\psi_i(u)|>M$ for any $i\geq0$. Thus, as $(\theta^{\alpha(a+1)}\tilde\psi_i(u_i))_i$ has been proven to be bounded, so must be $(u_i)_i$.

     To conclude on that sequence, we denote by $\rho$ a finite subsequential limit of $(u_i)_i$, if $\rho\ne u_m$ then deriving the limit of $\theta^{\alpha(a+1)}\tilde\psi_i(u_i)$ leads to 
        \begin{align}
           & -\frac{|\rho+1|^{1+\alpha}}{\alpha(\rho-a)+\rho+1} = |a+1|^\alpha \nonumber\\
           \text{i.e.}\quad & -|\rho+1|(|\rho+1|^\alpha-|a+1|^\alpha) = |a+1|^\alpha \alpha(\rho-a);
        \end{align}
    necessarily, $\rho=a$.
    It remains to prove that $\rho=u_m$ is absurd. Indeed, in this case the integrals $\int_{\Theta_i}\theta^{-\alpha a+u_i(1+\alpha)}d\theta$ converge either to $0$, either to $+\infty$. Therefore, that would make $(\theta_i^{\alpha(a+1)}\tilde\psi_i(u_i))_i$ converging either to $-\infty$, either to $0$, which in both case is different to $-|a+1|^\alpha$.

    Let us now work beyond the subset $U$ of $\RR$. First, if $u\in(-1,+\infty)$, the integrals that compose $\psi_i(u)$ both admit finite and positive limits when $i\to\infty$. The limit of $(|\psi_i(u)|)_i$ is moreover bounded from below as a consequence of Equation (\ref{eq:gendarmeJeffreys}):
        \begin{equation}\label{eq:minorationpsiu1infty}
            |\psi_i(u)|\geq \tilde K\frac{(u+1)^{1+\alpha}}{ \alpha(u-a)+u+1}\frac{1-\theta_i^{\alpha(u-a)+u+1}}{(1-\theta_i^{u+1})^{1+\alpha}} \geq \tilde K'|\log\theta_i|^{-1-\alpha}.  %
        \end{equation}
    Thus, there exists $i_2\geq0$ such that for any $i>i_2$ $|\psi_i(a)|<\tilde K'|\log\theta_i|^{-1-\alpha}$, consequently to $\psi_i(a)\equi{i\rightarrow\infty}|a+1|^{\alpha}\theta_i^{-\alpha(a+1)}\aseq{i\rightarrow\infty}o(|\log\theta_i|^{-\alpha-1})$.
    As a result, for any $i>i_2$:
        \begin{equation}
            \sup_{u\in(-1,+\infty)}\psi_i(u)<\psi_i(a)\leq\psi_i(u_i).
        \end{equation}
    Finally, if $u\in(u_m,-1)$, analogously than in Equation (\ref{eq:minorationpsiu1infty}), we can write
        \begin{equation}
            |\psi_i(u)|\geq\tilde K(u+1)^{1+\alpha}\frac{|\log\theta_i|}{(1-\theta_i^{u+1})^{1+\alpha}} \geq \tilde K''\theta_i^{-(u_m+1)(1+\alpha)}|\log\theta_i|^{-\alpha}.
        \end{equation}
    Once again, we have $\psi_i(a)\aseq{i\rightarrow\infty}o(\theta_i^{-(u_m+1)(1+\alpha)}|\log\theta_i|^{-\alpha})$ and we can consider $i_3\geq 0$ such that for any $i>i_3$:
        \begin{equation}
            \sup_{u\in(u_m,-1)}\psi_i(u)<\psi_i(a)\leq\psi_i(u_i).
        \end{equation}

    All the work that precedes proves that any sequence $(v_i)_i$ defined by $v_i\in\argmax_{\RR}\psi_i$ converges to $a$.

\paragraph{The case $\mathbf{a=-1}$}
In this case, we easily get that $\psi_i(a)\equi{i\rightarrow\infty}-|\log\theta_i|^{-\alpha}$. 
When $u<\eta<a$, 
    \begin{equation}
        |\psi_i(u)|\geq\tilde K\frac{|u+1|^{\alpha}}{\alpha+1} \frac{1-\theta_i^{-(\alpha+1)(u+1)}}{(1-\theta_i^{-u-1})^{1+\alpha}}\geq \hat K {|\eta+1|^{\alpha}}(1-\theta_i^{-(\eta-1)({1+\alpha})})
    \end{equation}
and when $u>\tilde\eta>a$,
    \begin{equation}
        |\psi_i(u)|\geq\tilde K\frac{|u+1|^{\alpha}}{\alpha+1} \frac{1-\theta_i^{(\alpha+1)(u+1)}}{(1-\theta_i^{u+1})^{1+\alpha}}\geq \hat K' {|\tilde\eta+1|^{\alpha}}(1-\theta_i^{(\tilde\eta-1)({1+\alpha})}).
    \end{equation}

Thus, for any $\eps>0$, the equations above with $\eta=a-\eps$ and $\tilde\eta=a+\eps$ let state that there exists an $i_1\geq0$ such that for any $i>i_1$:
    \begin{equation}
        \sup_{u\in(-\infty,a-\eps)}\psi_i(u)<\psi_i(a)\quad\text{and} \quad\sup_{u\in(a+\eps,\infty)}\psi_i(u)<\psi_i(a)
    \end{equation}
so that $\argmax_{\RR}\psi_i\subset(a-\eps,a+\eps)$, which let the definition of a sequence $(u_i)_i$ of maximal arguments of $\psi_i$ which converges to $a$.

\paragraph{Q-vague convergence}
The conclusion concerning the Q-vague convergence of the $\pi_i^\ast$ defined from the $u_i$ constructed in the work above: $\pi_i^\ast(\theta)\propto\theta^{u_i}$ is a direct result of \cite[Proposition 2.16]{Bioche2016}. Indeed, the convergence of sequence $(u_i)_i$ toward $a$, implies that the sequence of our priors converges Q-vaguely to $\pi^\ast$ such that $\pi^\ast(\theta)\propto{\theta^a}$.

\paragraph{Extension to other set $\Theta$}
We shall now demonstrate that the proven result extends itself to the general case: $\Theta=[c,b)$ or $(b,c]$, $b\in\RR\cup\{-\infty,\infty\}$.

We first consider $\Theta=(b,c]$ with $b\in\RR$.
We denote by $\cQ$ the class of priors after the substitution $\vartheta = (\theta-b)/(c-b)\in T=(0,1)$:  $\cQ=\{\tilde\pi(\vartheta) = (b-c)\pi(\vartheta (b-c)+b),\,\pi\in\hat\cP\}$.
Thus, $\cQ = \{\pi(\theta)\propto\vartheta^u,\,u\in\RR\}$.
We define the increasing sequence of compact sets $(\Theta_i)_i$ by $\Theta_i = [b+t_i,c]$ with $t_i\conv{i\rightarrow}0$. Therefore, for $\pi\in\hat\cP$, calling $\pi_i$ the re-normalized restriction of $\pi$ to $\Theta_i$ gives
    \begin{equation}
        l(\pi_i) = \tilde l(\tilde\pi_i) = C_\alpha\int_{T_i}\tilde\pi_i(\vartheta)^{1+\alpha}\tilde J(\vartheta)^{-\alpha}d\vartheta
    \end{equation}
with $T_i=[t_i,1]$, $\tilde\pi_i(\vartheta) = (b-c)\pi_i(\vartheta (b-c)+b)\in\cQ$ and $\tilde J(\vartheta) = (b-c)J(\vartheta (b-c)+b)\equi{\vartheta\rightarrow0}\vartheta^a$.
Thus, the work done above states that the family of maximal arguments $\tilde\pi_i^\ast$ of $\tilde l$ over $\cQ_i$  constitutes a sequence of priors that converges Q-vaguely toward $\tilde\pi^\ast(\vartheta)\propto\vartheta^a$.
Thus, the associate priors $\pi^\ast_i$ maximize $l$ over $\cP_i$ and constitute a sequence that converges Q-vaguely toward $\pi^\ast(\theta)\propto|\theta-b|^a$.

To treat the other cases, other substitutions with analogous work permit to conclude:
(i) the substitution $\vartheta=(b-\theta)/(b-c)$
when $\Theta=(c,b)$, $b\in\RR$; (ii) the substitution $\vartheta=1/(|\theta-c|+1)$ when 
$\Theta=[c,\infty)$ or $(-\infty,c]$.

\subsection{Proof of Proposition \ref{prop:constraints}}

Let us start by the uniqueness.
Recall the definition of an openly increasing sequence of compact sets $(\Theta_i)_{i\in\NN}$: there exist $i_0\geq0$ and a sequence $(V_i)_{i\geq i_0}$ of open subsets of $\Theta$  such that for any $i\geq i_0$
    \begin{equation}
        \Theta_i\subset V_i\subset \Theta_{i+1}.
    \end{equation}
This way, $\bigcup_iV_i=\Theta$ and the compacity of $\Theta$ imposes it to  be a finite union, so that $\Theta_i=\Theta$ for any $i\geq i_1$ for some $i_1\geq0$.\\
Thus, any $D_\alpha$-reference prior $\pi^\ast$ over a convex class such as $\tilde\cP$ must maximize $l$ as a prior on $\Theta$. The mapping $\pi\mapsto l(\pi)$ being strictly convex when $\pi$ is seen as a function in $\cR^\cC$, such maximal argument is unique.

Regarding the expression of $\pi^\ast$,
let us call $E$ the space of bounded a.e. continuous functions from $\Theta$ to $\RR$ and we equip $E$ with the supremum norm over $E$: $\|f\|=\sup_\Theta|f|$. The set $\Theta$ being supposed compact, the pair $(E,\|\cdot\|)$ constitutes a Banach vector space whose restriction $U$ composed by the positive functions of $E$ is an open and convex subset.
It is possible to see $l$ as being a concave function defined on $U$.

Let us compute the differentiate of $l$ over $U$. One can write $l=\phi_2\circ\phi_1$ with
\begin{equation}
    \phi_1:\pi\in E\longmapsto \pi^{1+\alpha}\in E ;\qquad \phi_2:\pi\in E\longmapsto C_\alpha \int_\Theta \pi(\theta)|\cI(\theta)|^{-\alpha/2}d\theta.
\end{equation}
As $\phi_2$ is a continuous linear mapping from $E$ to $\RR$, $l$ is differentiate while $\phi_1$ is, with for any $h\in E$:
    \begin{equation}
        dl(\pi) = \phi_2\circ d\phi_1(\pi).
    \end{equation}
Fix $\pi\in U$. For an $\eps>0$ there exists $\tilde\eps>0$ such that while $|x|<\|\pi\|$ and $|u|<\tilde\eps$ then $|(x+u)^{1+\alpha}-x^\alpha-(1+\alpha)x^\alpha u|<\eps|u|$. Thus for any $h\in E$ such that $\|h\|<\tilde\eps$, we have
    \begin{equation}
        \|\phi_1(\pi+h) -\phi_1(\pi)-(1+\alpha)\pi^\alpha h\|<\eps\|h\|.
    \end{equation}
We conclude that $\phi_1$ differentiable on $U$ with $d\phi_1(\pi)h = (1+\alpha)\pi^\alpha h$, for any $\pi\in U$, $h\in E$. This differentiate is additionally continuous, which makes $l$ continuously differentiable as well.

Thus, considering the additional constraint $\int_\Theta\pi(\theta)d\theta=1$, this problem can be treated applying the Lagrange multipliers theorem (see e.g. \citep{Zeidler}) to state that there exist $\lambda_0,\dots,\lambda_p\in\RR^{p+1}$ such that
    \begin{equation}
        dl(\pi^\ast)h - \lambda\int_\Theta h(\theta)d\theta - \sum_{i=1}^p\lambda_i\int_\Theta h(\theta)g_i(\theta)d\theta = 0
    \end{equation}
for any $h\in E$. Eventually, as $dl(\pi)h=C_\alpha(1+\alpha)\int_\Theta\pi(\theta)^{\alpha}|\cI(\theta)|^{-\alpha/2}h(\theta)d\theta$, we get
    \begin{align}
        \pi^\ast(\theta) = J(\theta)\bigg(\lambda_0+\sum_{i=1}^p\lambda_ig_i(\theta)\bigg)^{1/\alpha}.
    \end{align}

Moreover,
as $l$ is strictly concave and the constraints are linear, the second order condition for the Lagrangian states the reciprocal aspect of this result: calling $C$ the space of functions satisfying the constraints, if $\pi^\ast\in U\cap C$ satisfies the last equation above, it maximizes $l$.

To conclude, the elements in $\tilde\cP$ are associated to densities which belong to $E$ because $\Theta$ is compact.
They are non necessarily everywhere positive, but only supposed to be non-negative.
This set of bounded a.e. continuous and non-negative functions satisfying the constraints is actually the closure of $U\cap C$ over which $l$ is continuous.
Therefore, calling $\tilde\cP^\ast$ the class of positive priors in $\tilde\cP$, we have $\argmax_{\pi\in\tilde\cP}l(\pi)=\argmax_{\pi\in\tilde\cP^\ast}l(\pi)$ which is $\pi^\ast$.

\subsection{Proof of Theorem \ref{thm:lintoproper}}

Assuming that it exists, denote by $\pi^\ast$ the $D_\alpha$-reference prior over $\overline\cP$, denote as well
    \begin{equation}
        c = \int_{\Theta}\pi(\theta)g(\theta)d\theta,\quad Z_i = \int_{\Theta_i}\pi(\theta)d\theta
    \end{equation}
for any $i\in \NN$, considering an openly increasing sequence $(\Theta_i)_{i\in\NN}$ of compact sets that cover $\Theta$ over which $1>\pi^\ast(\Theta_i)>0$ for any $i$ such that $\Theta_i\ne\Theta$ (we show later that they exist).
In particular, $\pi^\ast$ must be the $D_\alpha$-reference prior over the class $\cP^c=\{\pi\in\cP,\,\int_\Theta\pi g =c\}\subset\cP$.

Let $i\in \NN$, if $\pi_i$ is a prior over $\Theta_i$ such that 
\begin{equation}
    \int_{\Theta_i}\pi_ig + \frac{1}{Z_i}\int_{\Theta\setminus\Theta_i}\pi^\ast g = \frac{c}{Z_i},
\end{equation}
then the prior $\pi$ on $\Theta$ defined by $\pi=Z_i\pi_i+\pi^\ast\indic_{\Theta\setminus\Theta_i}$ %
belongs to $\cP^c$.
Therefore, denoting $\pi^\ast_i$, the renormalized restriction of $\pi^\ast$ to $\Theta_i$, $l(\pi^\ast_i)$ must be larger than $l(\pi_i)$ on $\Theta_i$ by definition of the reference prior.

Thus, $\pi^\ast_i$ is a maximal argument of $l$ under the constraints 
    \begin{equation}
        \int_{\Theta_i}\pi_i = 1,\qquad \int_{\Theta_i}\pi_i g=\frac{c}{Z_i}-\frac{1}{Z_i}\int_{\Theta\setminus\Theta_i}\pi^\ast g.
    \end{equation}
Given the result of Proposition \ref{prop:constraints}, $\pi^\ast_i$ takes the form of:
    \begin{equation}
        \pi^\ast_i = J\cdot(\lambda^{(1)}_i+\lambda^{(2)}_ig)^{1/\alpha}
    \end{equation}
for some $\lambda_i^{(1)}$ and $\lambda_i^{(2)}$.

Then, denoting as well $\pi^\ast_{i+1}=J\cdot(\lambda^{(1)}_{i+1}+\lambda^{(2)}_{i+1}g)^{1/\alpha}$, and reminding that $\pi^\ast_{i+1}\propto\pi^\ast_i$ on $\Theta_i$, we deduce that $\lambda^{(1)}_{i+1}=\lambda^{(1)}_{i}$ as well as $\lambda^{(2)}_{i+1}=\lambda^{(2)}_{i}$.
Eventually, $\pi^\ast\propto J\cdot (\lambda_1+\lambda_2g)^{1/\alpha}$.
Thus, in the neighborhood of $b$, $\pi^\ast(\theta)\equi{\theta\rightarrow b}K\lambda_1^{1/\alpha}$ for some $K\ne0$ if $\lambda_1$ is non-null, which is discordant with the satisfaction of the constraint $\int_\Theta\pi^\ast f<\infty$ in the case where $\int_\Theta Jf=\infty$ or with the constraint $\int_\Theta\pi^\ast =1$ otherwise.

Finally, $\pi^\ast$ is proportional to $g^{1/\alpha}$ and the value of $c$ must be fixed by
\begin{equation}
    c = \left(\int_\Theta J\cdot g^{1+1/\alpha}\right)\cdot\left(\int_\Theta J\cdot g^{1/\alpha}\right)^{-1}.
\end{equation}

To finish the proof, we still have to show that $\pi^\ast$ is not null on any of the $\Theta_i$, or on any of the $\Theta\setminus\Theta_i$. First, there must exist $i_0$ such that $\pi^\ast(\Theta_{i_0})>0$, so that the $\pi^\ast(\Theta_i)>0$ for any $i\geq i_0$ and the sequence $(\Theta_i)_{i\geq i_0}$ can be considered instead of the initial one. Second, for any $i$, if $\Theta\setminus\Theta_i$ is non-empty then it has a non-empty interior, as a consequence of the definition of openly increasing sequences of compact sets. Thus, as $\pi^\ast$ is assumed to be positive, $\pi^\ast(\Theta\setminus\Theta_i)$ is non-null. Hence the result.

\subsection{Proof of Theorem \ref{thm:quasipostpropre}}

  Consider an openly increasing sequence of compact sets $(\Theta_i)_{i\in\NN}$ that covers $\Theta$. Let $i\in\NN$, for $\pi_i$ ---associated to its probability density on $\Theta_i$--- to be a reference prior over $\cP_i$, it must maximize the function $l$ under the constraints $\int_{\Theta_i}\pi_i=1$ and $\int_{\Theta_i}\pi_i g=c_i$. Thus, according to Proposition \ref{prop:constraints}, $\pi_i$ can be written:
     \begin{equation}
         \pi_i(\theta)\propto J(\theta)(\lambda_i^{(1)}+\lambda_i^{(2)}g(\theta))^{1/\alpha},
     \end{equation}
 for some $\lambda_i^{(1)}$ and $\lambda_i^{(2)}$.
 Considering, if required, a subsequence of $(\pi_i)_{i\in\NN}$, we can assume that $\lambda_i^{(1)}$ and $\lambda_i^{(2)}$ have constant signs.
 Also, by convexity we write:%
 \newcommand{\li}[1]{\lambda_i^{(#1)}}%
     \begin{align}\label{eq:convexityConsCpig}
         c_i^\alpha  & \geq \int_{\Theta_i}J(\theta)(\li1+\li2g(\theta))g(\theta)d\theta\left(\int_{\Theta_i}Jg\right)^{\alpha-1}  \\
         &\geq \li1\left(\int_{\Theta_i}Jg\right)^{\alpha} + \li2\frac{\int_{\Theta_i}Jg^{2} }{\left(\int_{\Theta_i}Jg\right)^{1-\alpha}},\nonumber
     \end{align}
    and
    \begin{align}\label{eq:convexityConspi1}
        1 & \geq \int_{\Theta_i}J(\theta)(\li1+\li2g(\theta))d\theta\left(\int_{\Theta_i}J\right)^{\alpha-1} \\
        &\geq \li1\left(\int_{\Theta_i}J\right)^\alpha + \li2 \frac{\int_{\Theta_i}Jg}{\left(\int_{\Theta_i}J\right)^{1-\alpha}}  . \nonumber
    \end{align}
We want to identify the possible subsequential limits of $\li1$. 
We assume that one exists in $(0,\infty]$, we can consider if required a subsequence to assume that ii is the actual the limit of $\li1$. 
This way, Equation (\ref{eq:convexityConspi1}) does not allow $\li2$ to be non-negative.
Thus, $\li2$ is negative and by convexity,
    \begin{equation}\label{eq:convixity2l2leq0}
        c_i^\alpha \geq \li1\left(\int_{\Theta_i} Jg\right)^\alpha +\li2 \left(\int_{\Theta_i} Jg^{1+1/\alpha}\right)^\alpha .
    \end{equation}
 Given that $(c_i)_i$ is bounded, and that $\int_{\Theta_i}Jg\conv{i\rightarrow\infty}\infty$ while $\int_{\Theta_i}Jg^{1+1/\alpha}$ has a finite limit, we deduce that $\li2$ diverges to $-\infty$.
 However, to ensure the expression of $\pi_i$ to be well defined, we must have $\li1+\li2g(\theta)\geq0$ for any $i$ and any $\theta$. As $g$ is non-null, it necessary comes that $(\li2/\li1)_i$ is a  bounded sequence.
Therefore, 
$\li2 \aseq{i\rightarrow\infty} o\left(\li1\left(\int_{\Theta_i}Jg\right)^\alpha\right) $
and the right hand side in Equation (\ref{eq:convixity2l2leq0}) diverges to $\infty$, which is a contradiction with $(c_i)_i$ being bounded.

As a consequence, any subsequential limit of $(\li1)_i$ must belong to $[-\infty,0]$. Without changing the notations, we can now assume that $(\li1)_i$ has a limit, which we suppose to be strictly negative firstly, also, $(\li1)_i$ can then be assumed to be always negative.
 Thus, still to ensure that $\li1+\li2g(\theta)\geq0$ for any $i$, $\theta$, and reminding that $g(\theta)\conv{\theta\rightarrow b}0$ we deduce $\frac{\li2}{|\li1|}\conv{i\rightarrow\infty}\infty$.
 Then, 
 the sequence of functions $\left(\pi_i/\li2\right)_i$ converges pointwisely to $\pi^\ast$, defined by $\pi^\ast(\theta)\propto J(\theta)g(\theta)^{1/\alpha}$ with  $|\frac{1}{\li2}\pi_i(\theta)|\leq C+g(\theta)$ for some $C$, with $g$ being continuous. %
According to \cite[Proposition 2.15]{Bioche2016}, what precedes implies that the sequence of priors $(\pi_i)_i$ converges Q-vaguely to $\pi^\ast$.

Eventually, if $(\li1)_i$ admits $0$ as a subsequential limit, 
we can assume that it is its actual limit.
Consequently,  the sequence of functions $(\pi_i/\li2)_i$ converges pointwisely to the density $\pi^\ast$ while being bounded  by $g$. Therefore, the sequence of priors $(\pi_i)_i$ converges Q-vaguely to $\pi^\ast$.
In any cases, $\pi^\ast$ is a quasi-reference prior over $\overline\cP'$.

\section{Conclusion}\label{sec:conclusion}

Prior elicitation remains an open topic in Bayesian analysis, with no solution satisfying all criteria (objectivity, computational feasibility, property, ...) simultaneously. 
In this work, we have focused on the objectivity and exploitability of priors.
Specifically, we have provided some solutions for practitioners seeking non-subjective priors, addressing the practical challenges of traditional reference priors, particularly in terms of complexity and proper aspect. Through the example we presented, we demonstrated that simple criteria can lead to the construction of reference priors that are either straightforward to formulate or proper, depending on the case.
Moreover, our results show that one cannot completely dispense with a consideration of the asymptotic properties of Jeffreys prior. This emphasizes the importance of studying this prior for a complete construction of objective priors.

Finally, this paper paves the way for several future research directions. Some studies might focus on the numerical computations of certain priors that we suggest, especially those derived from Theorem \ref{thm:lintoproper}, leveraging methods that already approximate usual reference priors in the literature, e.g. \citep{Berger2009,LeTriMinh2014}. Others might explore the flexible nature of the constraints initially imposed on the prior, using our formalism to define new constraints and reference priors tailored to their specific motivations.

\begin{acks}[Acknowledgments]
The author would like to thank his advisors Cyril Feau (Université Paris-Saclay, CEA) and Josselin Garnier (CMAP, CNRS, École polytechnique) for their advice and suggestions.
\end{acks}

\begin{acks}[Declaration of interest]
The author reports there are no competing interests to declare.
\end{acks}

\begin{acks}[Author's full addresses]
Correspondences to the author at: CMAP, CNRS, École polytechnique, Institut Polytechnique de Paris, 91120 Palaiseau, France; and at: Université Paris-Saclay, CEA, Service d'Études Mécaniques et Thermiques, 91191 Gif-sur-Yvette, France.
\end{acks}

%

\end{document}